\documentclass[11pt,leqno]{article}

\usepackage{amsmath}
\usepackage{amsthm}
\usepackage{amssymb}

\theoremstyle{plain}
\newtheorem{theorem}{Theorem}[section]
\newtheorem{lemma}[theorem]{Lemma}
\newtheorem{proposition}[theorem]{Proposition}

\theoremstyle{definition}
\newtheorem{definition}[theorem]{Definition}

\theoremstyle{remark}
\newtheorem{remark}[theorem]{Remark}

\numberwithin{equation}{section}

\begin{document}

\title{\textbf{The solution to the BCS gap equation \\
and the second-order phase transition \\
in superconductivity}}

\author{Shuji Watanabe\\
Division of Mathematical Sciences\\
Graduate School of Engineering, Gunma University\\
4-2 Aramaki-machi, Maebashi 371-8510, Japan\\
Email: shuwatanabe@gunma-u.ac.jp}

\date{}

\maketitle

\begin{abstract}
The existence and the uniqueness of the solution to the BCS gap equation of superconductivity is established in previous papers, but the temperature dependence of the solution is not discussed. In this paper, in order to show how the solution varies with the temperature, we first give another proof of the existence and the uniqueness of the solution and point out that the unique solution belongs to a certain set. Here this set depends on the temperature $T$. We define another certain subset of a Banach space consisting of continuous functions of both $T$ and $x$. Here, $x$ stands for the kinetic energy of an electron minus the chemical potential. Let the solution be approximated by an element of the subset of the Banach space above. We second show, under this approximation, that the transition to a superconducting state is a second-order phase transition. 

\medskip


\medskip

\end{abstract}


\section{Introduction and preliminaries}

We use the unit $k_B=1$, where $k_B$ stands for the Boltzmann constant. Let $\omega_D>0$ and $k\in\mathbb{R}^3$ stand for the Debye frequency and the wave vector of an electron, respectively. We denote Planck's constant by $h>0$ and set $\hslash=h/(2\pi)$. Let $m>0$ and $\mu>0$ stand for the electron mass and the chemical potential, respectively. We denote by $T (\geq 0)$ the temperature, and by $x$ the kinetic energy of an electron minus the chemical potential, i.e., $x=\hslash^2|k|^2/(2m)-\mu \in [ -\mu, \, \infty)$. Note that $0<\hslash\omega_D<<\mu$.

In the BCS model \cite{bcs, bogoliubov} of superconductivity, the solution to the BCS gap equation \eqref{eq:gapequation} below is called the gap function. We regard the gap function as a function of both $T$ and $x$, and denote it by $u$, i.e., $u: \, (T,\, x) \mapsto u(T,\, x)$ $(\geq 0)$. The BCS gap equation is the following nonlinear integral equation:
\begin{equation}\label{eq:gapequation}
u(T,\, x)=\int_{\varepsilon}^{\hslash\omega_D}
\frac{U(x,\,\xi)\, u(T,\,\xi)}{\,\sqrt{\,\xi^2+u(T,\,\xi)^2\,}\,}\,
\tanh \frac{\,\sqrt{\,\xi^2+u(T,\,\xi)^2\,}\,}{2T}\, d\xi,
\end{equation}
where $U(x,\,\xi)>0$ is the potential multiplied by the density of states per unit energy at the Fermi surface and is a function of $x$ and $\xi$. In \eqref{eq:gapequation} we introduce $\varepsilon>0$, which is small enough and fixed $(0<\varepsilon<<\hslash\omega_D)$. It is known that the BCS gap equation \eqref{eq:gapequation} is based on a superconducting state called the BCS state. In this connection, see \cite[(6.1)]{watanabe} for a new gap equation based on a superconducting state having a lower energy than the BCS state.

The integral with respect to $\xi$ in \eqref{eq:gapequation} is sometimes replaced by the integral over $\mathbb{R}^3$ with respect to the wave vector $k$. Odeh \cite{odeh}, and Billard and Fano \cite{billardfano} established the existence and the uniqueness of the positive solution to the BCS gap equation in the case $T=0$. In the case $T\geq 0$, Vansevenant \cite{vansevesant} determined the transition temperature (the critical temperature) and showed that there is a unique positive solution to the BCS gap equation. Recently, Hainzl, Hamza, Seiringer and Solovej \cite{hhss} proved that the existence of a positive solution to the BCS gap equation is equivalent to the existence of a negative eigenvalue of a certain linear operator to show the existence of a transition temperature. Hainzl and Seiringer \cite{haizlseiringer} also derived upper and lower bounds on the transition temperature and the energy gap for the BCS gap equation. Moreover, Frank, Hainzl, Naboko and Seiringer \cite{fhns} gave a rigorous analysis of the asymptotic behavior of the transition temperature at weak coupling.

However, the solution mentioned above belongs to a function space consisting of functions of the wave vector only, and the temperature is regarded as a parameter, i.e., in previous papers the existence and the uniqueness of the solution is established for $T$ fixed. So the temperature dependence of the solution is not discussed. For example, it is not shown that the solution is continuous for $T \geq 0$. Moreover, how the solution varies with the temperature is not studied. Studying the temperature dependence of the solution to the BCS gap equation is very important. This is because, by dealing with the thermodynamical potential, this study leads to the mathematical challenge of showing that the transition to a superconducting state is a second-order phase transition. So it is highly desirable to study the temperature dependence of the solution.

In this paper, in order to show how the solution varies with the temperature, we first give another proof of the existence and the uniqueness of the solution to the BCS gap equation \eqref{eq:gapequation}. More precisely, we show that the unique solution belongs to $V_T$ (see \eqref{eq:vt} below). Note that the set $V_T$ depends on $T$. We define a certain subset $W$ (see \eqref{eq:spacew} below) of a Banach space consisting of continuous functions of both $T$ and $x$. Let the solution be approximated by an element of $W$. We second show, under this approximation, that the transition to a superconducting state is a second-order phase transition. In other words, we show that the condition that the solution belongs to $W$ is a sufficient condition for the second-order phase transition in superconductivity.

In this connection, on the basis of the Banach fixed-point theorem, the author \cite[Theorem 2.3]{watanabe3} recently gave a proof of the statement that the solution to the BCS gap equation \eqref{eq:gapequation} is continuous for both  $T$ and $x$ when $T$ $(\geq 0)$ is small enough.

Let
\begin{equation}\label{eq:potentialu}
U(x,\,\xi)=U_1 \qquad \mbox{at all} \quad (x,\,\xi) \in [\varepsilon, \, \hslash\omega_D]^2,
\end{equation}
where $U_1>0$ is a constant. Then the gap function depends on the temperature $T$ only. We therefore denote the gap function by $\Delta_1$ in this case, i.e., $\Delta_1:\, T \mapsto \Delta_1(T)$. Then \eqref{eq:gapequation} leads to the simplified gap equation
\begin{equation}\label{eq:smplgapequation}
1=U_1\int_{\varepsilon}^{\hslash\omega_D}
 \frac{1}{\,\sqrt{\,\xi^2+\Delta_1(T)^2\,}\,}\,
 \tanh \frac{\, \sqrt{\,\xi^2+\Delta_1(T)^2\,}\,}{2T}\,d\xi.
\end{equation}
It is known that superconductivity occurs at temperatures below the transition temperature. The following is the definition of the transition temperature, which originates from the simplified gap equation \eqref{eq:smplgapequation}.
\begin{definition}[\cite{bcs}]\label{dfn:tcsmpl}
The transition temperature originating from the simplified gap equation \eqref{eq:smplgapequation} is the temperature $\tau_1>0$ satisfying
\[
1=U_1\int_{\varepsilon}^{\hslash\omega_D}
\frac{1}{\,\xi\,}\,\tanh \frac{\xi}{\,2\tau_1\,}\,d\xi.
\]
\end{definition}

The BCS model makes the assumption that there is a unique solution $\Delta_1: T \mapsto \Delta_1(T)$ to the simplified gap equation \eqref{eq:smplgapequation} and that it is of class $C^2$ with respect to the temperature $T$ (see e.g. \cite{bcs} and \cite[(11.45), p.392]{ziman}). The author \cite{watanabe2} has given a mathematical proof of this assumption on the basis of the implicit function theorem. Set
\begin{equation}\label{eq:delta0}
\Delta=\frac{\,
\sqrt{ \left( \hslash\omega_D-\varepsilon\, e^{1/U_1} \right)
\left( \hslash\omega_D-\varepsilon\, e^{-1/U_1} \right) }\,}
{\,\sinh\frac{1}{\,U_1\,}\,}.
\end{equation}

\begin{proposition}[{\cite[Proposition 2.2]{watanabe2}}]\label{prp:simplegap}
Let $\Delta$ be as in \eqref{eq:delta0}. Then there is a unique nonnegative solution $\Delta_1: [\,0,\,\tau_1\,] \to [0,\,\infty)$ to the simplified gap equation \eqref{eq:smplgapequation} such that the solution $\Delta_1$ is continuous and strictly decreasing on the closed interval $[\,0,\,\tau_1\,]$:
\[
\Delta_1(0)=\Delta>\Delta_1(T_1)>\Delta_1(T_2)>\Delta_1(\tau_1)=0, \qquad 0<T_1<T_2<\tau_1.
\]
Moreover, it is of class $C^2$ on the interval $[\,0,\,\tau_1\,)$ and satisfies
\[
\Delta_1'(0)=\Delta_1''(0)=0 \quad \mbox{and} \quad \lim_{T\uparrow \tau_1} \Delta_1'(T)=-\infty.
\]
\end{proposition}

\begin{remark}
We set $\Delta_1(T)=0$ for $T>\tau_1$.
\end{remark}

\begin{remark}
In Proposition \ref{prp:simplegap}, $\Delta_1(T)$ is nothing but $\sqrt{f(T)}$ in \cite[Proposition 2.2]{watanabe2}.
\end{remark}

Let $0<U_1<U_2$ , where $U_2>0$ is a constant. We assume the following condition on $U(\cdot,\,\cdot)$ throughout this paper:
\begin{equation}\label{eq:condition}
U_1 \leq U(x,\,\xi) \leq U_2 \quad \mbox{at all} \quad (x,\,\xi) \in [\varepsilon,\, \hslash\omega_D]^2, \qquad U(\cdot,\,\cdot) \in C^2([\varepsilon,\, \hslash\omega_D]^2).
\end{equation}
When $U(x,\,\xi)=U_2$ at all $(x,\,\xi) \in [\varepsilon,\, \hslash\omega_D]^2$, an argument similar to that in Proposition \ref{prp:simplegap} gives that there is a unique nonnegative solution $\Delta_2: [\,0,\,\tau_2\,] \to [0,\,\infty)$ to the simplified gap equation
\begin{equation}\label{eq:smplgapequation2}
1=U_2\int_{\varepsilon}^{\hslash\omega_D}
 \frac{1}{\,\sqrt{\,\xi^2+\Delta_2(T)^2\,}\,}\,
 \tanh \frac{\, \sqrt{\,\xi^2+\Delta_2(T)^2\,}\,}{2T}\,d\xi, \qquad
0\leq T\leq \tau_2.
\end{equation}
Here the transition temperature $\tau_2>0$ is defined by
\begin{equation}\label{dfn:tcsmpl2}
1=U_2\int_{\varepsilon}^{\hslash\omega_D}
\frac{1}{\,\xi\,}\,\tanh \frac{\xi}{\,2\tau_2\,}\,d\xi.
\end{equation}
We again set $\Delta_2(T)=0$ for $T>\tau_2$. A straightforward calculation gives the following.

\begin{lemma}\label{lm:tautau}\quad {\rm (a)} The inequality $\tau_1<\tau_2$ holds.

\noindent {\rm (b)} If \   $0\leq T<\tau_2$, then $\Delta_1(T)<\Delta_2(T)$. If \  $T\geq \tau_2$, then $\Delta_1(T)=\Delta_2(T)=0$.
\end{lemma}

On the basis of Proposition \ref{prp:simplegap}, the author \cite[Theorem 2.3]{watanabe2} proved that the transition to a superconducting state is a second-order phase transition under the restriction \eqref{eq:potentialu}.

As is mentioned above, we now introduce the thermodynamical potential $\Omega$ to study the phase transition in superconductivity. The thermodynamical potential $\Omega$ is of the form
\[
\Omega=-k_BT \,\ln Z,
\]
where $Z$ is the function called the partition function. For more details on the thermodynamical potential, see e.g. \cite[sec.~III]{bcs} or Niwa \cite[sec.~7.7.3]{niwa}. Let $N(x)\geq 0$ stand for the density of states per unit energy at the energy $x$ \quad $(-\mu \leq x < \infty)$ and set $N_0=N(0)>0$. Here, $N_0$ stands for the density of states per unit energy at the Fermi surface $(x=0)$. Note that the function $x \mapsto N(x)$ is continuous on $[-\mu , \, \infty)$. For the gap function $u$, set
\[
\Omega_S(T)=\Omega_N(T)+\Psi(T),
\]
where
\begin{eqnarray}\label{eq:omegan}
\Omega_N(T)&=&
-2N_0\int_{\displaystyle{\varepsilon}}^{\displaystyle{\hslash\omega_D}} x\,dx
-4N_0T\int_{\displaystyle{\varepsilon}}^{\displaystyle{\hslash\omega_D}}
 \ln\left( 1+e^{\displaystyle{-x/T}} \right)\,dx \\ \nonumber
& &+\Phi(T),\qquad T>0,
\end{eqnarray}
\begin{eqnarray}\label{eq:v}
\Phi(T)&=&
2\int_{\displaystyle{-\mu}}^{\displaystyle{-\hslash\omega_D}} x\, N(x)\,dx
 -2T\int_{\displaystyle{-\mu}}^{\displaystyle{-\hslash\omega_D}}
 N(x)\ln\left( 1+e^{\displaystyle{\,x/T}} \right)\,dx \\ \nonumber
& &-2T\int_{\displaystyle{\hslash\omega_D}}^{\infty} N(x)
 \ln\left( 1+e^{\displaystyle{-x/T}} \right)\,dx,\qquad T>0,
\end{eqnarray}
\begin{eqnarray}\label{eq:delta}
\Psi(T)&=&
-2N_0\int_{\displaystyle{\varepsilon}}^{\displaystyle{\hslash\omega_D}}
 \left\{ \sqrt{x^2+u(T,\,x)^2}-x\right\}\,dx  \\ \nonumber
& &+N_0\int_{\displaystyle{\varepsilon}}^{\displaystyle{\hslash\omega_D}}
\frac{u(T,\, x)^2}{\,\sqrt{\,x^2+u(T,\, x)^2\,}\,}\,
\tanh \frac{\,\sqrt{\,x^2+u(T,\, x)^2\,}\,}{2T}\, dx \\ \nonumber
& &-4N_0T
\int_{\displaystyle{\varepsilon}}^{\displaystyle{\hslash\omega_D}}
\ln \frac{1+e^{-\displaystyle{\sqrt{x^2+u(T,\, x)^2}/T}}}
{1+e^{-\displaystyle{x/T}}} \,dx,\quad 0<T\leq T_c\,.
\end{eqnarray}
Here, $T_c$ is the transition temperature defined by Definition \ref{dfn:tcgeneral} below and originates from the BCS gap equation \eqref{eq:gapequation}.

\begin{remark}
The integral $\displaystyle{ \int_{\displaystyle{\hslash\omega_D}}^{\infty} N(x) \ln\left( 1+e^{\displaystyle{-x/T}} \right)\,dx }$ on the right side of (\ref{eq:v}) is well defined for $T>0$, since $N(x)=O(\sqrt{x})$ as $x \to \infty$.
\end{remark}

\begin{definition}\label{dfn:thpo}
Let $\Omega_S(T)$ and $\Omega_N(T)$ be as above. The thermodynamical potential $\Omega$ in the BCS model is defined by (see e.g. Niwa \cite[sec.~7.7.3]{niwa})
\[
\Omega(T)=\left\{ \begin{array}{ll}\displaystyle{
 \Omega_S(T)} \qquad &(0<T\leq T_c),\\
\noalign{\vskip0.3cm} \displaystyle{
 \Omega_N(T)} \qquad &(T>T_c).
\end{array}\right.
\]
\end{definition}

\begin{remark}
Generally speaking, the thermodynamical potential $\Omega$ is a function of the temperature $T$, the chemical potential $\mu$ and the volume of our physical system. Fixing both $\mu$ and the volume of our physical system, we deal with the dependence of $\Omega$ on the temperature $T$ only.
\end{remark}

\begin{remark}
Hainzl, Hamza, Seiringer and Solovej \cite{hhss} studied the BCS gap equation with a more general potential examining the thermodynamic pressure.
\end{remark}

\begin{remark}\label{rmk:omegan-v}
It is shown in \cite[Lemmas 6.1 and 6.2]{watanabe2} that both of the functions $\Omega_N$ (see \eqref{eq:omegan}) and $\Phi$ (see \eqref{eq:v}), regarded as functions of $T$, are of class $C^2$ on $(0,\,\infty)$.
\end{remark}

\begin{definition}
We say that the transition to a superconducting state at the transition temperature $T_c$ is a second-order phase transition if the following conditions are fulfilled.

{\rm (a)}\   The thermodynamical potential $\Omega$, regarded as a function of $T$, is of class $C^1$ on $(0,\,\infty)$.

{\rm (b)}\   The thermodynamical potential $\Omega$, regarded as a function of $T$, is of class $C^2$ on $(0,\,\infty) \setminus \{ T_c\}$, and the second-order partial derivative $\left( \partial^2\Omega/\partial T^2\right)$ is discontinuous at $T=T_c$.
\end{definition}

\begin{remark}
As is known in condensed matter physics, condition (a) implies that the entropy $\displaystyle{ S=-\left( \partial\Omega/\partial T\right) }$ is continuous on $(0,\,\infty)$ and that, as a result, no latent heat is observed
at $T=T_c$ . On the other hand, (b) implies that the specific heat at constant volume, $\displaystyle{C_V=-T\left( \partial^2\Omega/\partial T^2\right)}$, is discontinuous at $T=T_c$ and that the gap $\Delta C_V$ in the specific heat at constant volume is observed at $T=T_c$ . For more details on the entropy well as the specific heat at constant volume, see e.g. \cite[sec.~III]{bcs} or Niwa \cite[sec.~7.7.3]{niwa}.
\end{remark}

The paper proceeds as follows. In section 2 we state our main results without proof. In sections 3, 4 and 5 we prove our main results.

\section{Main results}

Let $0 \leq T \leq \tau_2$ and fix $T$, where $\tau_2$ is that in \eqref{dfn:tcsmpl2}. We first consider the Banach space $C([\varepsilon,\, \hslash\omega_D])$ consisting of continuous functions of $x$ only, and deal with the following subset $V_T$:
\begin{equation}\label{eq:vt}
V_T=\left\{ u(T,\,\cdot) \in C([\varepsilon,\, \hslash\omega_D]): \; \Delta_1(T) \leq u(T,\,x) \leq \Delta_2(T) \;\; \mbox{at} \;\; x \in [\varepsilon,\, \hslash\omega_D] \right\}.
\end{equation}

\begin{remark}
The set $V_T$ depends on $T$. So we denote each element of $V_T$ by $u(T,\,\cdot)$.
\end{remark}

As is mentioned in the introduction, the existence and the uniqueness of the solution to the BCS gap equation was established in previous papers \cite{billardfano, fhns, hhss, haizlseiringer, odeh, vansevesant} and the uniqueness only holds for a nonnegative $U(\cdot,\,\cdot)$. However the temperature dependence of the solution is not discussed, and so we give another proof of the existence and the uniqueness of the solution to the BCS gap equation \eqref{eq:gapequation} so as to show how the solution varies with the temperature. More precisely, we show that for $T$ fixed, the unique solution belongs to $V_T$. Note that Proposition \ref{prp:simplegap} and Lemma \ref{lm:tautau} point out how $\Delta_1$ and $\Delta_2$ depend on the temperature and how $\Delta_1$ and $\Delta_2$ vary with the temperature.

\begin{theorem}\label{thm:gap}
Assume condition \eqref{eq:condition} on $U(\cdot,\,\cdot)$. Let $T \in [0,\, \tau_2]$ be fixed. Then there is a unique nonnegative solution $u_0(T,\,\cdot) \in V_T$ to the BCS gap equation \eqref{eq:gapequation}:
\[
u_0(T,\, x)=\int_{\varepsilon}^{\hslash\omega_D}
\frac{U(x,\,\xi)\, u_0(T,\, \xi)}{\,\sqrt{\,\xi^2+u_0(T,\, \xi)^2\,}\,}\,
\tanh \frac{\,\sqrt{\,\xi^2+u_0(T,\, \xi)^2\,}\,}{2T}\, d\xi, \quad
x \in [\varepsilon,\, \hslash\omega_D].
\]
Consequently, the solution is continuous with respect to $x$ and varies with the temperature as follows:
\[
\Delta_1(T) \leq u_0(T,\, x) \leq \Delta_2(T) \quad \mbox{at} \quad
(T,\,x) \in [0,\, \tau_2] \times [\varepsilon,\, \hslash\omega_D].
\]
\end{theorem}

\begin{remark}
In fact, Theorem \ref{thm:gap} holds true under
\[
U_1 \leq U(x,\,\xi) \leq U_2 \quad \mbox{at all} \quad (x,\,\xi) \in [\varepsilon,\, \hslash\omega_D]^2, \qquad U(\cdot,\,\cdot) \in C([\varepsilon,\, \hslash\omega_D]^2).
\]
But we assume condition \eqref{eq:condition} on $U(\cdot,\,\cdot)$ instead. This is because we deal with the subset $W$ (see \eqref{eq:spacew} below) so as to prove Theorem \ref{thm:phasetransition}.
\end{remark}

\begin{proposition}\label{prp:uzerotzero}
Let $T \in [\tau_1,\, \tau_2]$ be fixed and let $u_0(T,\,\cdot)$ be as in Theorem \ref{thm:gap}. If there is a point $x_1 \in [\varepsilon,\,\hslash\omega_D]$ satisfying $u_0(T,\,x_1)=0$, then $u_0(T,\,x)=0$ at all $x \in [\varepsilon,\,\hslash\omega_D]$.
\end{proposition}.

The existence of the transition temperature $T_c$ is pointed out in previous papers \cite{fhns, hhss, haizlseiringer, vansevesant}. In our case, it is defined as follows.

\begin{definition}\label{dfn:tcgeneral}\quad  Let $u_0(T,\,\cdot) \in V_T$ be as in Theorem \ref{thm:gap}. The transition temperature $T_c$ originating from the BCS gap equation \eqref{eq:gapequation} is defined by
\[
T_c=\inf \left\{ T>0:\, u_0(T,\, x)=0 \quad \mbox{at all} \; x \in [\varepsilon,\, \hslash\omega_D] \right\}.
\]
\end{definition}

\begin{remark}\label{rmk:tc1tctc2}
Combining Definition \ref{dfn:tcgeneral} with Theorem \ref{thm:gap} implies that $\tau_1\leq T_c \leq \tau_2$. For $T>T_c$, we set $u_0(T,\, x)=0$ at all $x \in [\varepsilon,\, \hslash\omega_D]$.
\end{remark}

\begin{proposition}\label{prp:u0delta}\quad Let $u_0(T,\,\cdot)$ be as in Theorem \ref{thm:gap}. If $U(x,\,\xi)=U_1$ at all $(x,\,\xi)\in [\varepsilon,\, \hslash\omega_D]^2$, then $u_0(T,\, x)=\Delta_1(T)$ and $T_c=\tau_1$.
\end{proposition}

We next consider the Banach space $C([0,\, T_c] \times [\varepsilon,\,\hslash\omega_D])$ consisting of continuous functions of both $T$ and $x$. Let us consider the following condition, which gives the behavior of functions as $T \to T_c$. We assume condition \eqref{eq:condition} on $U(\cdot,\,\cdot)$. Let $T_c$ be as in Definition \ref{dfn:tcgeneral} and let $\varepsilon_1>0$ be arbitrary.

\noindent \textbf{Condition (C).} \  For $u \in C([0,\, T_c] \times [\varepsilon,\, \hslash\omega_D]) \cap C^2\left( (0,\, T_c) \times [\varepsilon,\, \hslash\omega_D] \right)$, there are a unique $v \in C([\varepsilon,\, \hslash\omega_D])$ and a unique $w \in C([\varepsilon,\, \hslash\omega_D])$ satisfying the following.

(C1) \   $v(x)>0$ at all $x \in [\varepsilon,\, \hslash\omega_D]$.

(C2) \   For $\varepsilon_1>0$, there is a $\delta>0$ such that
$\left| T_c-T \right|<\delta$ implies
\[
\left| v(x)-\frac{\,u(T,\, x)^2\,}{\,T_c-T\,} \right|<T_c \,\varepsilon_1
\quad \mbox{and} \quad
\left| v(x)+2\,u(T,\, x)\frac{\partial u}{\,\partial T\,}(T,\, x) \right|<T_c \, \varepsilon_1\,,
\]
where $\delta$ does not depend on $x \in [\varepsilon,\, \hslash\omega_D]$.

(C3) \   Set $f(T,\, x)=u(T,\, x)^2$. Then, for $\varepsilon_1>0$, there is a $\delta>0$ such that $\left| T_c-T \right|<\delta$ implies
\[
\left| \frac{\, w(x)\,}{2}+
\frac{\,f(T,\, x)+(T_c-T)\,\displaystyle{ \frac{\,\partial f\,}{\,\partial T\,} }(T,\, x)\,}{\,(T_c-T)^2\,} \right|<\varepsilon_1 \quad \mbox{and} \quad
\left| w(x)-\frac{\,\partial^2 f\,}{\,\partial T^2\,}(T,\, x) \right|<\varepsilon_1\,,
\]
where $\delta$ does not depend on $x \in [\varepsilon,\, \hslash\omega_D]$.
\begin{remark}\label{rmk:utc} \   If $u \in C([0,\, T_c] \times [\varepsilon,\, \hslash\omega_D]) \cap C^2\left( (0,\, T_c) \times [\varepsilon,\, \hslash\omega_D] \right)$ satisfies condition (C), then \  $\displaystyle{ u(T_c\,,\, x)=0 }$ at all $x \in [\varepsilon,\, \hslash\omega_D]$.
\end{remark}

We deal with the following subset $W$ of the Banach space $C([0,\, T_c] \times [\varepsilon,\,\hslash\omega_D])$. Dealing with $W$ is important both in studying smoothness of the thermodynamical potential with respsect to $T$ and in showing that the transition to a superconducting state is a second-order phase transition.
\begin{eqnarray}\label{eq:spacew}
& & \qquad \\ \nonumber
W&=&\left\{ u \in C([0,\, T_c] \times [\varepsilon,\,\hslash\omega_D]) \cap C^2((0,\, T_c) \times [\varepsilon,\,\hslash\omega_D]):
\Delta_1(T)\leq u(T,\,x) \leq \Delta_2(T) \right. \\ \nonumber
& & \quad \left. \;
\mbox{at}\; (T,\,x)\in [0,\, T_c] \times [\varepsilon,\,\hslash\omega_D], \quad u \; \mbox{satisfies condition (C)} \right\}.
\end{eqnarray}

\begin{remark}
Let $u \in W$. Then, for $T \geq T_c$, we set $u(T,\,x)=0$ at all $x \in [\varepsilon,\,\hslash\omega_D]$. 
\end{remark}

The set $W$ is not empty. Let $U_3 > 0$ be a constant satisfying
\[
1=U_3\int_{\varepsilon}^{\hslash\omega_D}
\frac{1}{\,\xi\,}\,\tanh \frac{\xi}{\,2T_c\,}\,d\xi.
\]
Note that $(0<)\, U_1\leq U_3\leq U_2$ (see Remark \ref{rmk:tc1tctc2}). An argument similar to that in Proposition \ref{prp:simplegap} gives that there is a unique nonnegative solution $\Delta_3: [\,0,\,T_c\,] \to [0,\,\infty)$ to the simplified gap equation
\[
1=U_3\int_{\varepsilon}^{\hslash\omega_D}
 \frac{1}{\,\sqrt{\,\xi^2+\Delta_3(T)^2\,}\,}\,
 \tanh \frac{\, \sqrt{\,\xi^2+\Delta_3(T)^2\,}\,}{2T}\,d\xi.
\]
Indeed, the function $\Delta_3$ just above is an element of $W$ (see \cite{watanabe2}). Therefore, $W \not= \emptyset$.

Define a mapping $A$ by
\begin{equation}\label{eq:mapping}
Au(T,\, x)=\int_{\varepsilon}^{\hslash\omega_D}
\frac{U(x,\,\xi)\, u(T,\, \xi)}{\,\sqrt{\,\xi^2+u(T,\, \xi)^2\,}\,}\,
\tanh \frac{\,\sqrt{\,\xi^2+u(T,\, \xi)^2\,}\,}{2T}\, d\xi,\qquad
u \in W.
\end{equation}

\begin{proposition}\label{prp:uauu0} \   Assume condition \eqref{eq:condition} on $U(\cdot,\,\cdot)$. Let $W$, a subset  of the Banach space $C([0,\, T_c] \times [\varepsilon,\,\hslash\omega_D])$, be as above. Let $u_0(T,\,\cdot) \in V_T$ be as in Theorem \ref{thm:gap}.

{\rm (a)}\   The mapping $\displaystyle{ A: \, W \longrightarrow W }$ is continuous with respect to the norm of the Banach space $C([0,\, T_c] \times [\varepsilon,\,\hslash\omega_D])$.

{\rm (b)}\   Let $u \in W$. Let $0 \leq T \leq \tau_2$ and fix $T$. Then all of $u(T,\, \cdot)$, $Au(T,\, \cdot)$ and $u_0(T,\,\cdot)$ belong to $V_T$. Consequently, at all $(T,\,x) \in [0,\, \tau_2] \times [\varepsilon,\,\hslash\omega_D]$,
\[
\Delta_1(T) \leq u(T,\,x), \, Au(T,\, x), \, u_0(T,\, x) \leq \Delta_2(T).
\]
\end{proposition}

We choose $U_1$ and $U_2$ (see \eqref{eq:condition}) such that the following inequality holds:
\begin{equation}\label{eq:inequality}
\sup_{0 \leq T \leq \tau_2} \left| \Delta_2(T)-\Delta_1(T) \right| < \varepsilon_2,
\end{equation}
where $\varepsilon_2>0$ is small enough. Then it follows from Proposition \ref{prp:uauu0} (b) that for $u \in W$,
\begin{equation}\label{eq:gapfunction}
\left| u(T,\,x)-u_0(T,\, x) \right|< \varepsilon_2, \quad
\left| Au(T,\, x)-u_0(T,\, x) \right|< \varepsilon_2, \quad
\left| Au(T,\,x)-u(T,\, x) \right|< \varepsilon_2.
\end{equation}
at all $(T,\,x) \in [0,\, \tau_2] \times [\varepsilon,\,\hslash\omega_D]$.

\noindent \textbf{Approximation (A).} \   The gap function on the right side of \eqref{eq:delta} is the solution $u_0$ of Theorem \ref{thm:gap}, i.e., the solution to the BCS gap equation \eqref{eq:gapequation}. But no one gives the proof of the statement that there is a unique solution in $W$ to the BCS gap equation \eqref{eq:gapequation}. In view of \eqref{eq:gapfunction}, we then let $u_0$ be approximated by a $u \in W$, and replace the gap function on the right side of \eqref{eq:delta} by this $u \in W$.

Let $g: \, [0,\, \infty) \rightarrow \mathbb{R}$ be given by
\begin{equation}\label{eq:fng}
g(\eta)= \left\{ \begin{array}{ll}\displaystyle{
\frac{1}{\,\eta^2\,}\left( \frac{1}{\,\cosh^2\eta \,}
 -\frac{\,\tanh\eta\,}{\eta}\right) } \qquad &(\eta>0),\\
\noalign{\vskip0.3cm} \displaystyle{
-\frac{\,2\,}{\,3\,} } &(\eta=0).
\end{array}\right.
\end{equation}
Note that $g(\eta)<0$. See Lemma \ref{lm:gproperty} below for some properties of $g$.

\begin{theorem}\label{thm:phasetransition} \   Assume condition \eqref{eq:condition} on $U(\cdot,\,\cdot)$. Let $U_1$ and $U_2$ be chosen such that \eqref{eq:inequality} holds, and let the solution $u_0$ of Theorem \ref{thm:gap} be approximated by a $u \in W$ as stated in approximation (A) above. Let $v \in C([\varepsilon,\,\hslash\omega_D])$ be as in condition (C). Then the following hold.

{\rm (a)}\   The transition to a superconducting state at the transition temperature $T_c$ is a second-order phase transition. Consequently, the condition that the solution to the BCS gap equation \eqref{eq:gapequation} belongs to $W$ is a sufficient condition for the second-order phase transition in superconductivity.

{\rm (b)}\   The gap $\Delta C_V$ in the specific heat at constant volume at the transition temperature $T_c$ is given by the form
\[
\Delta C_V=-\frac{N_0}{\, 8\, T_c\,}
 \int_{\varepsilon/(2\, T_c)}^{\hslash\omega_D/(2\, T_c)} v(2T_c\,\eta)^2
 \, g(\eta)\, d\eta \quad (>0).
\]
\end{theorem}

\begin{remark}
Suppose that $U(x,\,\xi)=U_1$ at all $(x, \,\xi) \in [\varepsilon,\, \hslash\omega_D]^2$. By Proposition \ref{prp:u0delta}, $u_0(T,\, x)=\Delta_1(T)$ and $T_c=\tau_1$. Therefore the function $\Delta_1$ of Proposition \ref{prp:simplegap} becomes an element of $W$ (see \cite{watanabe2}), and hence, $u$ of Theorem \ref{thm:phasetransition} can be replaced by  $\Delta_1$. So, setting $\varepsilon=0$, we find that the gap $\Delta C_V$ reduces to the form (see \cite[Proposition 2.4]{watanabe2})
\begin{equation}\label{eq:smplform}
\Delta C_V=-N_0f'(T_c)\tanh \frac{\,\hslash\omega_D\,}{2T_c} \quad (>0),
\end{equation}
where \  $\displaystyle{ f'(T_c)=-\lim_{T \uparrow T_c}
\frac{\,\Delta_1(T)^2\,}{\,T_c-T\,} }$.
\end{remark}

\section{Proof of Theorem \ref{thm:gap}}

Let $0 \leq T \leq \tau_2$ and fix $T$. Let $V_T$ be as in \eqref{eq:vt}. A straightforward calculation gives the following.

\begin{lemma}
The set $V_T$ is bounded, closed and convex.
\end{lemma}

Let $x \in [\varepsilon,\,\hslash\omega_D]$ and define a mapping $B$ by
\begin{equation}\label{eq:mappingb}
Bu(T,\, x)=\int_{\varepsilon}^{\hslash\omega_D}
\frac{U(x,\,\xi)\, u(T,\, \xi)}{\,\sqrt{\,\xi^2+u(T,\, \xi)^2\,}\,}\,
\tanh \frac{\,\sqrt{\,\xi^2+u(T,\, \xi)^2\,}\,}{2T}\, d\xi,\qquad
u(T,\,\cdot) \in V_T \, .
\end{equation}

\begin{lemma}\label{lm:aucd} \quad $\displaystyle{ Bu(T,\,\cdot) \in C([\varepsilon,\,\hslash\omega_D]) }$ for $u(T,\,\cdot) \in V_T$.
\end{lemma}

\begin{proof} \  For $\varepsilon_1>0$, let \  
$\displaystyle{ \delta=
\frac{\varepsilon_1}{\, \displaystyle{ \hslash\omega_D\,\sup_{(x,\,\xi) \in [\varepsilon,\,\hslash\omega_D]^2} \left| \frac{\,\partial U\,}{\partial x}(x,\,\xi) \right| }
\,} }$. Then $\left| \, x-x_0\,\right|<\delta$ implies
\begin{eqnarray*}
\left| \, Bu(T,\, x)-Bu(T,\, x_0)\,\right|
&\leq& \int_{\varepsilon}^{\hslash\omega_D}
\left| U(x,\,\xi)-U(x_0,\,\xi) \right| \, d\xi \cr
&\leq& \hslash\omega_D\,\left| \, x-x_0\,\right|
 \sup_{(x,\,\xi) \in [\varepsilon,\,\hslash\omega_D]^2}
 \left| \frac{\,\partial U\,}{\partial x}(x,\,\xi) \right| \cr
&<& \varepsilon_1\,.
\end{eqnarray*}
\end{proof}

Since $\delta$ in the proof just above does not depend on $u(T,\,\cdot) \in V_T$, we immediately have the following.

\begin{lemma}\label{lm:equic}
The set $\displaystyle{ BV_T=\left\{ Bu(T,\,\cdot):\; u(T,\,\cdot) \in V_T \right\} }$  $\left( \subset C([\varepsilon,\,\hslash\omega_D]) \right)$ is equicontinuous.
\end{lemma}

\begin{lemma}\label{lm:d1aud2} \quad Let $u(T,\,\cdot) \in V_T$. Then \quad$\displaystyle{ \Delta_1(T) \leq Bu(T,\, x) \leq \Delta_2(T) }$ at $x \in [\varepsilon,\,\hslash\omega_D]$.
\end{lemma}

\begin{proof}\quad We show $Bu(T,\, x) \leq \Delta_2(T)$.
Since
\[
\frac{u(T,\, \xi)}{\,\sqrt{\,\xi^2+u(T,\, \xi)^2\,}\,}\leq
\frac{\Delta_2(T)}{\,\sqrt{\,\xi^2+\Delta_2(T)^2\,}\,}\,,
\]
it follows from \eqref{eq:smplgapequation2} that
\begin{eqnarray*}
Bu(T,\, x) \leq \int_{\varepsilon}^{\hslash\omega_D}
\frac{U_2\, \Delta_2(T)}{\,\sqrt{\,\xi^2+\Delta_2(T)^2\,}\,}\,
\tanh \frac{\,\sqrt{\,\xi^2+\Delta_2(T)^2\,}\,}{2T}\, d\xi
=\Delta_2(T).
\end{eqnarray*}
The rest can be shown similarly by \eqref{eq:smplgapequation}.
\end{proof}

Combining Lemma \ref{lm:aucd} with Lemma \ref{lm:d1aud2} immediately implies the following.

\begin{lemma} \quad $\displaystyle{BV_T \subset V_T}$.
\end{lemma}

By Lemma \ref{lm:d1aud2}, the set $BV_T$ is uniformly bounded since
\[
Bu(T,\, x) \leq \Delta_2(0)=\frac{\,
\sqrt{ \left( \hslash\omega_D-\varepsilon\, e^{1/U_2} \right)
\left( \hslash\omega_D-\varepsilon\, e^{-1/U_2} \right) }\,}
{\,\sinh\frac{1}{\,U_2\,}\,} \qquad \mbox{for} \quad u(T,\,\cdot) \in V_T.
\]
Combining Lemma \ref{lm:equic} with the Ascoli--Arzel$\grave{\mbox{a}}$ theorem thus yields the following.

\begin{lemma}
The set $BV_T$, a subset of the Banach space $C([\varepsilon,\,\hslash\omega_D])$, is relatively compact.
\end{lemma}

\begin{lemma}\label{lm:Bconti} \   The mapping $\displaystyle{ B: \, V_T \longrightarrow V_T }$ is continuous with respect to the norm of the Banach space $C([\varepsilon,\,\hslash\omega_D])$.
\end{lemma}

\begin{proof} \   Let $u(T,\,\cdot),\,v(T,\,\cdot) \in V_T$. Then
\begin{eqnarray*}
& & \left| \frac{u(T,\, \xi)}{\,\sqrt{\,\xi^2+u(T,\, \xi)^2\,}\,}\,
 \tanh \frac{\,\sqrt{\,\xi^2+u(T,\, \xi)^2\,}\,}{2T}
 -\frac{v(T,\, \xi)}{\,\sqrt{\,\xi^2+v(T,\, \xi)^2\,}\,}\,
 \tanh \frac{\,\sqrt{\,\xi^2+v(T,\, \xi)^2\,}\,}{2T} \right| \cr
\noalign{\vskip2mm}
& & \leq \frac{\,\left| u(T,\, \xi)-v(T,\, \xi)\right|\,}
 {\,\sqrt{\,\xi^2+u(T,\, \xi)^2\,}\,}
 +v(T,\, \xi)\left| \frac{1}{\,\sqrt{\,\xi^2+u(T,\, \xi)^2\,}\,}
    -\frac{1}{\,\sqrt{\,\xi^2+v(T,\, \xi)^2\,}\,} \right| \cr
\noalign{\vskip2mm}
& & \quad +\frac{v(T,\, \xi)}{\,\sqrt{\,\xi^2+v(T,\, \xi)^2\,}\,}\,
   \left| \tanh \frac{\,\sqrt{\,\xi^2+u(T,\, \xi)^2\,}\,}{2T}
   -\tanh \frac{\,\sqrt{\,\xi^2+v(T,\, \xi)^2\,}\,}{2T} \right| \cr
\noalign{\vskip2mm}
& & \leq 3\frac{\,\left| u(T,\, \xi)-v(T,\, \xi)\right|\,}{\xi}\,.
\end{eqnarray*}
Thus \  $\displaystyle{
\left\| \, Bu(T,\,\cdot)-Bv(T,\,\cdot)\,\right\| \leq 3\, U_2
\ln \frac{\,\hslash\omega_D\,}{\,\varepsilon\,}\cdot
\left\| \, u(T,\,\cdot)-v(T,\,\cdot)\,\right\| }$, \quad
where $\displaystyle{ \left\| \cdot \right\| }$ stands for the norm of the Banach space $C([\varepsilon,\,\hslash\omega_D])$.
\end{proof}

We now have the following.

\begin{lemma}\quad The mapping $\displaystyle{ B: \, V_T \longrightarrow V_T}$ is compact, i.e., the mapping $\displaystyle{ B: \, V_T \longrightarrow V_T}$ is continuous and transforms bounded sets into relatively compact sets.
\end{lemma}

See Zeidler \cite[pp.39--40]{zeidler} for (nonlinear) compact operators. The Schauder fixed-point theorem (see e.g. Zeidler \cite[p.61]{zeidler}) thus implies the following.

\begin{lemma}\label{lm:u0}\quad Let $T \in [0,\, \tau_2]$ be fixed. Then the mapping $\displaystyle{ B: \, V_T \longrightarrow V_T}$ has at least one fixed point $u_0(T,\,\cdot) \in V_T$, i.e.,
\[
u_0(T,\,\cdot)=Bu_0(T,\,\cdot)\,,\qquad u_0(T,\,\cdot) \in V_T.
\]
\end{lemma}

Let us prove the uniqueness of $u_0(T,\,\cdot) \in V_T$.

\begin{lemma}\label{lm:unique} \quad Let $T \in [0,\, \tau_2]$ be fixed. Then the mapping $\displaystyle{ B: \, V_T \longrightarrow V_T}$ has a unique fixed point $u_0(T,\,\cdot) \in V_T$.
\end{lemma}

\begin{proof}
We give a proof similar to that of Thereom 24.2 given by Amann \cite{amann}. Let $v_0(T,\,\cdot) \in V_T$ be another fixed point of $B$, i.e., $v_0(T,\,\cdot)=Bv_0(T,\,\cdot)$. We deal with the case where $u_0(T,\, x)>0$ and $v_0(T,\, x)>0$ at all $x \in [\varepsilon,\,\hslash\omega_D]$ (see Proposition \ref{prp:uzerotzero}).

\textit{Step 1. The case where $\{ x \in [\varepsilon,\,\hslash\omega_D]: u_0(T,\, x) \geq v_0(T,\, x) \} \not= \emptyset$ and $\{ x \in [\varepsilon,\,\hslash\omega_D]: u_0(T,\, x) < v_0(T,\, x) \} \not= \emptyset$.}

Then there are a number $t$ \  $(0<t<1)$ and a point $x_0 \in [\varepsilon,\,\hslash\omega_D]$ such that
\begin{equation}\label{eq:T_0x_0}
u_0(T,\, x)\geq t\,v_0(T,\, x) \quad (x \in [\varepsilon,\,\hslash\omega_D]) \quad \mbox{and} \quad u_0(T,\, x_0)=t\, v_0(T,\, x_0).
\end{equation}
Hence
\begin{eqnarray*}
u_0(T,\, x_0) &=& \int_{\varepsilon}^{\hslash\omega_D}
 \frac{U(x_0,\,\xi)\, u_0(T,\, \xi)}{\,\sqrt{\,\xi^2+u_0(T,\, \xi)^2\,}\,}
 \,\tanh \frac{\,\sqrt{\,\xi^2+u_0(T,\, \xi)^2\,}\,}{2T}\, d\xi \cr
& & \cr
&\geq& \int_{\varepsilon}^{\hslash\omega_D}
 \frac{U(x_0,\,\xi)\, t\,v_0(T,\, \xi)}{\,\sqrt{\,\xi^2+t^2\,v_0(T,\, \xi)^2 \,}\,}\,
 \tanh \frac{\,\sqrt{\,\xi^2+t^2\,v_0(T,\, \xi)^2\,}\,}{2T}\, d\xi \cr
& & \cr
&>& t\,\int_{\varepsilon}^{\hslash\omega_D}
\frac{U(x_0,\,\xi)\, v_0(T,\, \xi)}{\,\sqrt{\,\xi^2+v_0(T,\, \xi)^2\,}\,}\,\tanh \frac{\,\sqrt{\,\xi^2+v_0(T,\, \xi)^2\,}\,}{2T}\, d\xi \cr
& & \cr
&=& t\,v_0(T,\, x_0),
\end{eqnarray*}
which contradicts \eqref{eq:T_0x_0}.

\textit{Step 2. The case where $u_0(T,\, x) \leq v_0(T,\, x)$ at all $x \in [\varepsilon,\,\hslash\omega_D]$.}

We again have \eqref{eq:T_0x_0}. Hence the same reasoning applies.

Thus $u_0(T,\,\cdot)=v_0(T,\,\cdot)$.
\end{proof}

The proof of Lemma \ref{lm:unique} is based on the condition that $U(\cdot,\,\cdot)$ is nonnegative. The proof of Theorem \ref{thm:gap} is complete.

We assumed also, in the proof of Lemma \ref{lm:unique}, that $u_0(T,\, x)>0$ and $v_0(T,\, x)>0$ at all $x \in [\varepsilon,\,\hslash\omega_D]$. We now show that if there is a point $x_1 \in [\varepsilon,\,\hslash\omega_D]$ satisfying $u_0(T,\,x_1)=0$, then $u_0(T,\,x)=0$ at all $x \in [\varepsilon,\,\hslash\omega_D]$ (see Proposition \ref{prp:uzerotzero}).

\noindent \textit{Proof of Proposition \ref{prp:uzerotzero}}.

Since $u_0(T,\,x_1)=0$, Theorem \ref{thm:gap} implies
\[
\int_{\varepsilon}^{\hslash\omega_D}
 \frac{U(x_1,\,\xi)\, u_0(T,\, \xi)}{\,\sqrt{\,\xi^2+u_0(T,\, \xi)^2\,}\,}
 \,\tanh \frac{\,\sqrt{\,\xi^2+u_0(T,\, \xi)^2\,}\,}{2T}\, d\xi=0.
\]
Since the integrand is nonnegative, it follows that $u_0(T,\, \xi)=0$ at all $\xi \in [\varepsilon,\,\hslash\omega_D]$.

\begin{remark}\quad Let $\tau_1 \leq T < \tau_2$ and fix $T$. Clearly, $0 \in V_T$ is a fixed point of $B$. If, for such a $T$, there are the two fixed points $u_0(T,\,\cdot) \in V_T$ mentioned in Lemma \ref{lm:unique} and $0 \in V_T$, then the fixed point $0 \in V_T$ is disregarded.
\end{remark}

Let us give a proof of Proposition \ref{prp:u0delta}.

\noindent \textit{Proof of Proposition \ref{prp:u0delta}}.

By Theorem \ref{thm:gap},
\[
u_0(T,\, x)=U_1 \int_{\varepsilon}^{\hslash\omega_D}
 \frac{u_0(T,\, \xi)}{\,\sqrt{\,\xi^2+u_0(T,\, \xi)^2\,}\,}
 \,\tanh \frac{\,\sqrt{\,\xi^2+u_0(T,\, \xi)^2\,}\,}{2T}\, d\xi, \qquad
 x \in D.
\]
Hence $u_0(T,\, x)$ does not depend on $x$. Therefore we denote $u_0(T,\,\cdot)$ by $u_0(T)$. Then
\[
u_0(T)\left\{ 1-U_1\int_{\varepsilon}^{\hslash\omega_D}
 \frac{1}{\,\sqrt{\,\xi^2+u_0(T)^2\,}\,}
 \,\tanh \frac{\,\sqrt{\,\xi^2+u_0(T)^2\,}\,}{2T}\, d\xi \right\}=0.
\]
Since $u_0(T) \not= 0$, it follows from \eqref{eq:smplgapequation} and Proposition \ref{prp:simplegap} that $u_0(T)=\Delta_1(T)$ and $T_c=\tau_1$.

\section{Proof of Proposition \ref{prp:uauu0}}

Let $A$ be as in \eqref{eq:mapping}. A straightforward calculation gives the following.

\begin{lemma} \   Let $u\in W$. Then $\displaystyle{ Au \in C([0,\, T_c] \times [\varepsilon,\,\hslash\omega_D]) \cap C^2( (0,\, T_c) \times [\varepsilon,\,\hslash\omega_D] ) }$.
\end{lemma}

A proof similar to that of Lemma \ref{lm:d1aud2} immediately gives the following.
\begin{lemma} \   Let $u\in W$. Then
\[
\Delta_1(T) \leq Au(T,\, x) \leq \Delta_2(T) \quad \mbox{at} \quad
(T,\,x) \in [0,\, T_c] \times [\varepsilon,\,\hslash\omega_D].
\]
\end{lemma}

We next show that $Au$ \   $(u \in W)$ satisfies condition (C) above. For $u \in W$, let $v$ be as in condition (C). Set
\begin{equation}\label{eq:functionF}
F(x)=\left( \int_{\varepsilon}^{\hslash\omega_D}
\frac{\, U(x,\,\xi)\,\sqrt{v(\xi)}\,}{\xi}
\tanh \frac{\xi}{\,2T_c\,}\, d\xi \right)^2 \quad (>0), \quad
\varepsilon \leq x \leq \hslash\omega_D\,.
\end{equation}

A straightforward calculation gives the following.

\begin{lemma}
Let $F$ be as in \eqref{eq:functionF}. Then $F \in C([\varepsilon,\,\hslash\omega_D])$.
\end{lemma}

\begin{lemma}
Let $F$ be as in \eqref{eq:functionF}. Then, for $\varepsilon_1>0$, there is a $\delta>0$ such that $\left| T_c-T \right|<\delta$ implies
\[
\left| F(x)-\frac{\,\{Au(T,\, x)\}^2\,}{\,T_c-T\,} \right|<T_c \,
\varepsilon_1\,,
\]
where $\delta$ does not depend on $x$ \  $(\varepsilon \leq x \leq \hslash\omega_D)$.
\end{lemma}

\begin{proof} \   Set $\displaystyle{M=\sup_{\varepsilon \leq \xi \leq \hslash\omega_D} v(\xi)}$. Let $\varepsilon_1$ and $\delta$ be as in (C2) of condition (C), and let $0<\varepsilon_1<1$. Then,
\begin{eqnarray*}
& & \; \left| F(x)-\frac{\,\{Au(T,\, x)\}^2\,}{\,T_c-T\,} \right| \cr
&\leq& \left| \int_{\varepsilon}^{\hslash\omega_D}
\frac{\, U_2\,\sqrt{v(\xi)}\,}{\xi}\, d\xi
+\int_{\varepsilon}^{\hslash\omega_D}
\frac{\, U_2\,}{\xi}\sqrt{\frac{\,u(T,\,\xi)^2\,}{\,T_c-T\,}}
\, d\xi \right| \cdot \left| I_1(x)+I_2(x)+I_3(x) \right| \cr
&\leq& 2U_2\sqrt{M+T_c}\,\ln\frac{\,\hslash\omega_D\,}{\varepsilon} \cdot
\left| I_1(x)+I_2(x)+I_3(x) \right|,
\end{eqnarray*}
where
\begin{eqnarray*}
I_1(x) &=& \int_{\varepsilon}^{\hslash\omega_D}
\frac{\,U(x,\,\xi)\,}{\xi}\left(
\sqrt{v(\xi)}-\sqrt{ \frac{\,u(T,\,\xi)^2\,}{\,T_c-T\,} }
\right) \tanh \frac{\xi}{\,2T_c\,}\,d\xi,\cr
I_2(x) &=& \int_{\varepsilon}^{\hslash\omega_D}
U(x,\,\xi)\sqrt{ \frac{\,u(T,\,\xi)^2\,}{\,T_c-T\,} } \left(
\frac{1}{\,\xi\,}-\frac{1}{\,\sqrt{\xi^2+u(T,\,\xi)^2}\,}\right)
\tanh \frac{\xi}{\,2T_c\,}\,d\xi,\cr
I_3(x) &=& \int_{\varepsilon}^{\hslash\omega_D}
\frac{U(x,\,\xi)}{\,\sqrt{\xi^2+u(T,\,\xi)^2}\,}
\sqrt{ \frac{\,u(T,\,\xi)^2\,}{\,T_c-T\,} } \left(
\tanh \frac{\xi}{\,2T_c\,}
-\tanh \frac{\,\sqrt{\xi^2+u(T,\,\xi)^2}\,}{\,2T\,} \right)\,d\xi.
\end{eqnarray*}
Let $T>T_c/2$. Then, by (C2), $\left| T_c-T \right|<\delta$ implies
\[
\left| I_1(x) \right| \leq \varepsilon_1\frac{\,U_2\,}{2}
\int_{\varepsilon}^{\hslash\omega_D} \frac{d\xi}{\,\sqrt{v(\xi)}\,}.
\]
Note that $\delta$ does not depend on $x$ \  $(\varepsilon \leq x \leq \hslash\omega_D)$. Since
\[
u(T,\,\xi)^2=\frac{\, u(T,\,\xi)^2\,}{T_c-T}\left( T_c-T \right),
\]
$u(T,\,\xi)$ tends to 0 uniformly with respect to $\xi$ as $T \to T_c$ by (C2). Therefore, $\left| T_c-T \right|<\delta$ implies $u(T,\,\xi) < T_c \, \varepsilon_1$. Recalling $T>T_c/2$ and (C2), we can show similarly that
\[
\left| I_2(x) \right| \leq \varepsilon_1\frac{\,U_2\, T_c\sqrt{M+T_c}\,}{\varepsilon} \quad \hbox{and} \quad
\left| I_3(x) \right| \leq \varepsilon_1\, U_2\left( \frac{\,\sqrt{M+T_c}\,}{2}\ln\frac{\,\hslash\omega_D\,}{\varepsilon}+\frac{\, \hslash\omega_D \,}{\, \sqrt{2T_c} \,} \right).
\]
Note also that $\varepsilon >0$ is fixed.
\end{proof}

A straightforward calculation similar to the above gives the following.

\begin{lemma}
Let $F$ be as in \eqref{eq:functionF}. Then, for $\varepsilon_1>0$, there is a $\delta>0$ such that $\left| T_c-T \right|<\delta$ implies
\[
\left| F(x)+2\, Au(T,\, x)\frac{\partial Au}{\,\partial T\,}(T,\, x)
\right|<T_c \, \varepsilon_1\,,
\]
where $\delta$ does not depend on $x$ \  $(\varepsilon \leq x \leq \hslash\omega_D)$.
\end{lemma}

For $u \in W$, let $v$ and $w$ be as in condition (C). For $\varepsilon \leq x \leq \hslash\omega_D$, set
\begin{eqnarray}\label{eq:functionG}
& &G(x)=\int_{\varepsilon}^{\hslash\omega_D}
\frac{\, U(x,\,\xi)\,\sqrt{v(\xi)}\,}{\xi}
\tanh \frac{\xi}{\,2T_c\,}\, d\xi
\int_{\varepsilon}^{\hslash\omega_D} U(x,\,\eta) \times \\
&\times& \left\{ \left( \frac{w(\eta)}{\,\eta\sqrt{v(\eta)}\,}-
\frac{\,2\sqrt{v(\eta)^3}\,}{\eta^3} \right)
\tanh \frac{\eta}{\,2T_c\,}+\frac{\sqrt{v(\eta)}}{\,\cosh^2
\displaystyle{ \frac{\eta}{\,2T_c\,}\,} }
\left( \frac{v(\eta)}{\,T_c\,\eta^2\,}+\frac{2}{\,T_c^2\,} \right)
\right\}\, d\eta.\nonumber
\end{eqnarray}

A straightforward calculation similar to the above gives the following.

\begin{lemma}\   Let $G$ be as in \eqref{eq:functionG}. Then $G \in C([\varepsilon,\,\hslash\omega_D])$.
\end{lemma}

\begin{lemma}\   Let $G$ be as in \eqref{eq:functionG}. Set $H(T,\, x)=\{ Au(T,\, x)\}^2$. Then, for $\varepsilon_1>0$, there is a $\delta>0$
such that $\left| T_c-T \right|<\delta$ implies
\[
\left| \frac{\, G(x)\,}{2}+\frac{\, H(T,\, x)+(T_c-T)\,\displaystyle{
\frac{\,\partial H\,}{\,\partial T\,} }(T,\, x)\,}{\,(T_c-T)^2\,} \right|<\varepsilon_1 \quad \mbox{and} \quad
\left| G(x)-\frac{\,\partial^2 H\,}{\,\partial T^2\,}(T,\, x) \right|<\varepsilon_1\,,
\]
where $\delta$ does not depend on $x$ \  $(\varepsilon \leq x \leq \hslash\omega_D)$.
\end{lemma}

We therefore have $AW \subset W$. A proof similar to that of Lemma \ref{lm:Bconti} immediately gives the following.
\begin{lemma}\   The mapping $\displaystyle{
A: \, W \longrightarrow W}$ is continuous with respect to the norm of the Banach space $C([0,\, T_c] \times [\varepsilon,\,\hslash\omega_D])$.
\end{lemma}

We have thus proved (a) of Proposition \ref{prp:uauu0}. Part (b) follows immediately from (a).

\section{Proof of Theorem \ref{thm:phasetransition}}

Let $T_1>0$ be arbitrary. 

\begin{lemma}\label{lm:deltat}
Let the function $\Psi$ be as in \eqref{eq:delta}. Then the function $\Psi$ is differentiable on $(T_1\,,\,T_c]$, and $(\partial \Psi/\partial T)(T_c)=0$.
\end{lemma}

\begin{proof}
A straightforward calculation gives that $\Psi$ is differentiable on $(T_1\,,\,T_c)$. So it suffices to show that $\Psi$ is differentiable at $T=T_c$ and $(\partial \Psi/\partial T)(T_c)=0$. Note that $\Psi(T_c)=0$ since $u(T_c\,,\,x)=0$ at all $x \in [\varepsilon,\, \hslash\omega_D]$ (see Remark \ref{rmk:utc}). Then, for $T<T_c$,
\begin{eqnarray}\label{eq:deltadifferentiable}
\,\quad -\frac{\Psi(T)}{\,T_c-T\,}&=&
\frac{2N_0}{\,T_c-T\,}
\int_{\displaystyle{\varepsilon}}^{\displaystyle{\hslash\omega_D}}
 \left\{ \sqrt{x^2+u(T,\,x)^2}-x\right\}\,dx  \\ \nonumber
& &-\frac{N_0}{\,T_c-T\,}
\int_{\displaystyle{\varepsilon}}^{\displaystyle{\hslash\omega_D}}
\frac{u(T,\, x)^2}{\,\sqrt{\,x^2+u(T,\, x)^2\,}\,}\,
\tanh \frac{\,\sqrt{\,x^2+u(T,\, x)^2\,}\,}{2T}\, dx \\ \nonumber
& &+\frac{4N_0T}{\,T_c-T\,}
\int_{\displaystyle{\varepsilon}}^{\displaystyle{\hslash\omega_D}}
\ln \frac{1+e^{-\displaystyle{\sqrt{x^2+u(T,\, x)^2}/T}}}
{1+e^{-\displaystyle{x/T}}} \,dx.
\end{eqnarray}
By (C2) in condition (C),
\[
\left| \frac{\,\sqrt{x^2+u(T,\,x)^2}-x\,}{\,T_c-T\,} \right| \leq
\frac{\,v(x)+T_c\,}{x},
\]
where $\varepsilon_1<1$ is assumed. Hence, the Lebesgue dominated convergence theorem implies that as $T \uparrow T_c$ , the first term on the right side of \eqref{eq:deltadifferentiable} tends to
\[
\lim_{T \uparrow T_c} \frac{2N_0}{\,T_c-T\,}
\int_{\displaystyle{\varepsilon}}^{\displaystyle{\hslash\omega_D}}
\left\{ \sqrt{x^2+u(T,\,x)^2}-x\right\}\,dx=
N_0 \int_{\displaystyle{\varepsilon}}^{\displaystyle{\hslash\omega_D}}
\frac{\, v(x)\,}{x}\,dx.
\]
Similarly we can deal with the second and the third terms. Thus $\Psi$ is differentiable at $T=T_c$ and
\[
-\lim_{T \uparrow T_c} \frac{\Psi(T)}{\,T_c-T\,}=0.
\]
\end{proof}

A straightforward calculation gives the following.

\begin{lemma}\label{lm:gproperty}
Let $g$ be as in \eqref{eq:fng}. Then the function $g$ is of class $C^1$
on $[0,\,\infty)$ and satisfies
\[
g(\eta)<0 \quad (\eta \geq 0),\qquad g'(0)=0,\qquad
\lim_{\eta\to\infty}g(\eta)=\lim_{\eta\to\infty}g'(\eta)=0.
\]
\end{lemma}

\begin{lemma}\label{lm:deltac2}
Let $\Psi$ be as in Lemma \ref{lm:deltat}. Then $\partial \Psi/\partial T$ is differentiable on $(T_1\,,\,T_c]$, and
\[
\frac{\,\partial^2\Psi\,}{\,\partial T^2\,}(T_c)=\frac{N_0}{\,8T_c^2\,}
\int_{\varepsilon/(2T_c)}^{\hslash\omega_D/(2T_c)}
v(2T_c\,\eta)^2 \, g(\eta)\,d\eta \quad (<0).
\]
\end{lemma}

\begin{proof}
A straightforward calculation gives that $\partial \Psi/\partial T$ is differentiable on $(T_1\,,\,T_c)$. So it suffices to discuss its differentiability at $T=T_c$. Note that $(\partial \Psi/\partial T)(T_c)=0$ by Lemma \ref{lm:deltat}. Then
\begin{eqnarray*}
-\lim_{T \uparrow T_c} \frac{\,\displaystyle{
\frac{\,\partial \Psi\,}{\partial T}(T)}\,}{\,T_c-T\,}&=&-N_0
\lim_{T \uparrow T_c}
\int_{\displaystyle{\varepsilon}}^{\displaystyle{\hslash\omega_D}}
\frac{\, u(T,\,x)^2\,}{\, T_c-T\,}\; u(T,\,x)
\frac{\,\partial u(T,\,x)\,}{\,\partial T\,}\frac{1}{\,x^2+u(T,\,x)^2\,}
\times \\ \noalign{\vskip0.3cm} & &\quad \times
\left( \frac{1}{\,2T\cosh^2 \frac{\,\sqrt{\,x^2+u(T,\, x)^2\,}\,}{2T}\,}
-\frac{\, \tanh \frac{\,\sqrt{\,x^2+u(T,\, x)^2\,}\,}{2T}\,}{\sqrt{\,x^2+u(T,\, x)^2\,}} \right)\,dx.
\end{eqnarray*}
Note that $u(T_c\,,\,x)=0$ at all $x \in [\varepsilon,\, \hslash\omega_D]$ by Remark \ref{rmk:utc}. Hence, combining the Lebesgue dominated convergence theorem with (C2) in condition (C) proves the lemma.
\end{proof}

Note that $T_1>0$ is arbitrary. Therefore, combining Lemmas \ref{lm:deltat} and \ref{lm:deltac2} with Remark \ref{rmk:omegan-v} immediately implies the following.

\begin{lemma}
The thermodynamical potential $\Omega$, regarded as a function of $T$, is of class $C^1$ on $(0,\,\infty)$.
\end{lemma}

A straightforward calculation based on (C3) in condition (C) gives the following.

\begin{lemma}\  The function $\partial^2 \Psi/\partial T^2$ is continuous on $(0,\,T_c]$.
\end{lemma}

Lemma \ref{lm:deltac2} immediately implies the following.

\begin{lemma}\  The thermodynamical potential $\Omega$, regarded as a function of $T$, is of class $C^2$ on $(0,\,\infty) \setminus \{ T_c\}$. Consequently, the transition to a superconducting state at the transition temperature $T_c$ is a second-order phase transition, and hence the condition that the solution to the BCS gap equation \eqref{eq:gapequation} belongs to $W$ is a sufficient condition for the second-order phase transition in superconductivity.
\end{lemma}

Since the specific heat at constant volume is of the form $\displaystyle{C_V=-T\left( \partial^2\Omega/\partial T^2\right)}$, the gap $\Delta C_V$ in $C_V$ at the transition temperature $T_c$ is given by (see Definition \ref{dfn:thpo})
\[
\Delta C_V=-T_c \left\{ \lim_{T\uparrow T_c}
\frac{\,\partial^2\Omega\,}{\,\partial T^2\,}(T)-
\lim_{T\downarrow T_c} \frac{\,\partial^2\Omega\,}{\,\partial T^2\,}(T)
\right\}=-T_c\,\frac{\,\partial^2\Psi\,}{\,\partial T^2\,}(T_c).
\]

We thus have the following by Lemma \ref{lm:deltac2}.

\begin{lemma}
The gap $\Delta C_V$ in the specific heat at constant volume at the transition temperature $T_c$ is given by the form
\[
\Delta C_V=-\frac{N_0}{\, 8\, T_c\,}
 \int_{\varepsilon/(2\, T_c)}^{\hslash\omega_D/(2\, T_c)} v(2T_c\,\eta)^2
 \, g(\eta)\, d\eta \quad (>0).
\]
\end{lemma}

\noindent \textbf{Acknowledgments}

S. Watanabe is supported in part by the JSPS Grant-in-Aid for Scientific Research (C) 21540110.


\end{document}